\documentclass[aps,pra,twocolumn,nofootinbib,superscriptaddress,showpacs]{revtex4}
\usepackage{amsthm}
\usepackage [latin1]{inputenc}
\usepackage[scr=zapfc,scrscaled=1.15]{mathalfa}
\usepackage{color}
\usepackage{amsmath}
\usepackage{amsfonts}
\usepackage{amssymb}
\newcommand{\blu}{\color{blue}}
\newtheorem{theorem1}{Theorem}
\newtheorem{definition1}[theorem1]{Definition}
\newtheorem{prop}[theorem1]{Proposition}

\begin{document}

\title{Entropic Bounds For Unitary Testers and Mutually Unbiased Unitary Bases}

\author{Jesni Shamsul Shaari}
\address{Faculty of Science, International Islamic University Malaysia (IIUM),
Jalan Sultan Ahmad Shah, Bandar Indera Mahkota, 25200 Kuantan, Pahang, Malaysia}
\address{Institute of Mathematical Research (INSPEM), University Putra Malaysia, 43400 UPM Serdang, Selangor, Malaysia.}
\author{Stefano Mancini}
\address{School of Science \& Technology, University of Camerino, I-62032 Camerino, Italy}
\address{ INFN Sezione di Perugia, I-06123 Perugia, Italy}

\date{\today}

\begin{abstract}
We define the entropic bounds, i.e minimal uncertainty for pairs of \textit{unitary testers} in distinguishing between unitary transformations not unlike the well known entropic bounds for observables. We show that in the case of specific sets of testers which pairwise saturate the trivial zero bound, the testers are all equivalent in the sense their statistics are the same. On the other hand, when maximal bounds are saturated by such sets of testers, the unitary operators would form unitary bases which are mutually unbiased. This resembles very much the role of mutually unbiased bases in maximizing the entropic bounds for observables. We show how such a bound can be useful in certain quantum cryptographic protocols.
\end{abstract}


\maketitle
\section{Introduction}
Measurements of observables in classical physics is aimed at an objective result regarding an inherent property of a system and in principle the various observables can be measured with arbitrary precision. In the case of physics on the other hand, there are bounds on the amount of information one can have when measuring two incompatible observables. Originally denoted as `Heisenberg's Uncertainty Principle' and later made rigorous by Robertson, a  desirable formulation, in the context of entropic bounds was first given by Deutsch \cite{D} and later refined by Maassen and Uffink \cite{maassen}. The entropic bounds for a pair of incompatible measurements, in the context of `preparation uncertainty' reflects the impossibility of predicting both outcomes when making such measurements \cite{step}.
For a review on the subject, we refer to \cite{step,coles}.

While such a bound has received a fair amount of treatment, issues related to determining unknown quantum processes is another matter entirely. Determining an unknown quantum transformation requires not so much observables, but procedures or strategies that would provide us with information regarding the nature of the transformation. It involves the preparation of some state, which is subjected to the unknown transformation and finally measured to determine how it has evolved. Such strategies can be formalized in the context of `process positive-operator-valued measure' (PPOVM) or `quantum testers' \cite{ziman,sedlak}. Quantum testers can be understood as a generalization of positive-operator-valued measures (POVM). While the latter is a set of operators (positive and summing to the identity) used to determine the probability for selected outcomes of measurements in relation to quantum states, the former is a set of operators used to determine the probability for outcomes of process determination.

Moving from uncertainties in measuring observables for quantum states, in the case of pairs of testers, we ask if there would be some bounds on the uncertainty in the results of two differing testing of some transformation. This can be seen almost as the direct consequence of lower bounding entropies in the case of observables to one involving quantum testers.

More precisely, we will frame the problem within the context of guessing games inline with the notion of preparation uncertainty \cite{coles}. In short, preparation uncertainty for observables can be understood as follows: a party, would prepare and distribute quantum states to another who would measure in either one of two bases. Informed of the bases chosen, the party who prepared the state would have to guess the result of the measurement. Incompatibility of measurement bases would imply that the results cannot be predicted with certainty. Hence, analogously for our study, we propose the incompatibility of testers to reflect the inability of predicting the outcomes when testing some transformation.
 We should note that our take on incompatibility here is different from that in \cite{sedlak}; which discusses the issue of `joint testing' i.e. 
 testers are incompatible if the statistics of their outcomes cannot be reproduced when taking the marginals of the probability distribution of the outcomes for a third joint tester. This can be seen as the analogy to the issue of `measurement uncertainty'.
 
In this work, we shall consider testing procedures which use bipartite quantum states (comprising of two $d$-dimensional quantum states that may be entangled) as input to test unitary transformations and measurements thereafter of the transformed states. The measurements made, as commonly in practise, are projective measurements. Unitary transformations are, in principle, how ideal (closed) quantum systems actually evolve and thus represent a fundamental interest. Despite its simplicity, we discuss some interesting implications. In particular we show the possibility of singling out features of the unitary operators tested by looking at pairs of testers which saturates minimal or maximal entropic bound. This has applications in more fields like quantum cryptography \cite{LM}.

\section{Prelude: Mutually Unbiased Basis}

Before starting, it is instructive to briefly surf on a related matter; i.e mutually unbiased basis (MUB) \cite{durt}. Consider two orthonormal bases, $\{|\alpha_0\rangle,...,|\alpha_{d-1}\rangle\}$ and $\{|\beta_0\rangle,...,|\beta_{d-1}\rangle\}$ for $\mathcal{H}_d$. The two bases are called MUB if
\begin{eqnarray}
|\langle\alpha_i|\beta_j\rangle|=1/d, \forall i,j,=0,\ldots, d-1
\end{eqnarray}
thus reflecting the equiprobable transition between states in one basis to another. The notion of MUB plays an important role in establishing strong uncertainty relations \cite{step,coles}. As a matter of fact, when measurements of any quantum system is made in the basis $\mathbb{V}=\{|\mathbb{V}_i\rangle\langle\mathbb{V}_i|\}$ and $\mathbb{W}=\{|\mathbb{W}_i\rangle\langle\mathbb{W}_i|\}$, the well known uncertainty relation 
\begin{eqnarray}\label{ents}
H(\mathbb{V})+H(\mathbb{W})\ge -\log{\max_{i,j}|\langle \mathbb{V}_i|\mathbb{W}_j\rangle}|^2
\end{eqnarray}
holds true, with $H(\cdot)$ the Shannon entropy of the outcomes. The bound is saturated if and only if $\{|\mathbb{V}_0\rangle,...,|\mathbb{V}_{d-1}\rangle\}$ and $\{|\mathbb{W}_0\rangle,...,|\mathbb{W}_{d-1}\rangle\}$ are MUB.

In a complete analogy to MUB for Hilbert spaces, ref.\cite{jssarxiv} studied the notion of mutually unbiased unitary basis (MUUB) for $M(d,\mathbb{C})$.
Two distinct orthogonal basis, $\{P_0,...,P_{D-1}\}$ and $\{Q_0,...,Q_{D-1}\}$ comprising of unitary transformations for some $D$-dimensional subspace of the vector space $M(d,\mathbb{C})$ are sets of MUUB if 
\begin{eqnarray}\label{MUUBdef}
\left|\text{Tr}(P_i^{\dagger}Q_j)\right |^2=\kappa~~,~~\forall i,j,=0,\ldots, D-1
\end{eqnarray}
for some constant, $\kappa$, which takes on value $1$ and $d$ for a $d^2$ and $d$ dimensional subspace of $M(d,\mathbb{C})$ respectively \cite{jssarxiv,arxiv}. The definition for MUUB was motivated by the fidelity definition of ref.\cite{acin}, which can be written as $|\text{Tr}(u^\dagger u_g)|^2/{d^2}$; which describes how well one unitary operator $u$ compares to a guess $u_g$ for it, given a single use of the operation.

\section{Testers}

\noindent Consider an unknown process, $\mathcal{E}$, that acts on a $d$-dimensional quantum system. To gain any information regarding the nature of the process, one can submit a $d$-dimensional quantum system, as a probe, (possibly entangled with some ancilla) to be acted upon by that process. A measurement can subsequently be made thereafter on the whole (including ancilla) quantum system. Mathematically, such a setting can be captured in the context of quantum testers. The following quick description of quantum testers is very much derived from ref.\cite{sedlak} and is referred to for a more detailed treatment.

Let the process $\mathcal{E}$ be a completely positive trace-nonincreasing linear map that maps operators on the Hilbert space $\mathcal{H}^{(a)}$ to that of $\mathcal{H}^{(b)}$ (superscripts are used as labels so as not to confuse with subscripts denoting dimensionality). The above setting can rigorously be noted as a pair, $\mathscr{T}=(\rho,\{P_i\})$, where $\rho$ is the probe-ancilla initial density operator on the Hilbert space $\mathcal{H}^{(a)}\otimes\mathcal{H}^{(anc)}$ with $\mathcal{H}^{(anc)}$ as the Hilbert space for the ancilla and $\{P_i\}$ is a POVM on  $\mathcal{H}^{(b)}\otimes\mathcal{H}^{(anc)}$, the output-ancilla space.
The probability, $p_k$, for an outcome $k$ is thus given by 
\begin{eqnarray}\label{sedlaks1}
p_k=\text{Tr}[P_k(\mathcal{E}\otimes \mathcal{I}^{(anc)})(\rho)].
\end{eqnarray}
With $|\Psi\rangle=\sum_{i=0}^{d-1}|i\rangle |i\rangle\in\mathcal{H}^{(a)}\otimes\mathcal{H}^{(a)}$ as an unnormalised maximally entangled state and $E:=(\mathcal{E}\otimes \mathcal{I})(|\Psi\rangle\langle\Psi |)$ the Choi operator for $\mathcal{E}$, eq.(\ref{sedlaks1}) can be written as 
\begin{eqnarray}\label{sedlaks2}
p_k=\text{Tr}[T_kE],
\end{eqnarray}
 with 
\begin{eqnarray}T_k=\text{Tr}_{anc}[(P_k\otimes I^{(a)})(I^{(b)}\otimes S\rho^\text{t} S)].
\end{eqnarray}
Here $S$ is the SWAP operator and $\rho^\text{t}$, the partial transpose of $\rho$ on $\mathcal{H}_a$. A quantum tester (or PPOVM) is then defined as the set $\{T_i\}$ with the conditions being the positivity of $T_i$ as well as $\sum_i T_i=I^{(b)}\otimes [\text{Tr}_{anc}(\rho)]^\text{t}$.
\footnote{It is worth noting that the quantum tester is referred more specifically as a quantum 2-tester in \cite{bisioaps}.}

However, for our intent and purposes in what follows, we will not make explicit use of the operator structure in the definition above. Rather we will refer to the main ingredients of the quantum testers; namely the input as well as the measurement operators. To avoid confusion in nomenclature, we  will just use the term \textit{tester} as opposed to quantum tester or PPOVM. We will further restrict our work to unitary testers, i.e. testers that are used to test unitary transformations with pure states for input (written as ket states) and projective measurements at the end. Let us thus provide the formal definition of testers.
\begin{definition1}\label{defT}
Consider an unitary operator,  $u\in\mathcal{U}(d)$, to be determined (tested), where $\mathcal{U}(d)$ is the unitary group acting on $\mathcal{H}_d$. A tester, $\mathscr{T}=(|\psi\rangle,\{|\chi_i\rangle\langle\chi_i|\})$, is a setup consisting of a pure state $|\psi\rangle\in\mathcal{H}_{d}\otimes \mathcal{H}_{d}$ as an input to undergo the tested unitary transformation, $u\in\mathcal{U}(d)$, and a set of orthogonal measurement projectors $(\{|\chi_i\rangle\langle\chi_i|\}$ where $\sum|\chi_i\rangle\langle\chi_i|=I_{d}\otimes I_{d})$ to subsequently measure the transformed state.
\end{definition1}
We note that in our definition, the ancilla for the input state is also a $d$-dimensional quantum system and the unitary to be tested would be acting on only on one half of a bipartite input state (not necessarily separable).
It would make sense to consider testers designed to allow one to distinguish between unitary operators in a certain set, $S=\{u_1\otimes I_d, ..., u_{D}\otimes I_d \}$, where $D=d,d^2$ depending on the number of projection measurement operators summing up to the identity $I_{d}\otimes I_{d}$. The case for $d$ can be understood as having completely separable bipartite states for input (ancilla-free \cite{ziman}) and the measurement operators would only project onto states defining a $d$-dimensional subspace of $\mathcal{H}_{d}\otimes \mathcal{H}_{d}$.  

\section{Uncertainties and Entropic Bounds}
It is instructive to define properly what we mean when we wish to identify entropic bounds for a pair of testers. Motivated by the guessing game, inline with the notion of preparation uncertainty (mentioned earlier) reflecting the uncertainty principle for pairs of observables, we define similarly. Let a party, Ailin, prepare a large number of identical unitary transformations, $u$ from some subset acting on the Hilbert space $\mathcal{H}_d$ and submit to another, Boris, for testing. Boris may choose either one of two testers, $\mathscr{T}_1=(|\psi\rangle,\{|\chi_i\rangle\langle\chi_i|\})$ and $\mathscr{T}_2=(|\phi\rangle,\{|\zeta_i\rangle\langle\zeta_i|\})$
and thereafter informs Ailin of the choice of testers.

Let the results of a tester $\mathscr{T}$ be described by a sample space of $\{t_i~|~i=1,..,n\}$ with assigned probabilities $p_{i},...p_{n}$ respectively. The entropic bound for the pair of testers $\mathscr{T}_1$ and $\mathscr{T}_2$ when testing the transformation $u$, is the smallest value for a number $c$, such that 
\begin{eqnarray}
H(\mathscr{T}_1,u)+H(\mathscr{T}_2,u)\ge c~,~ \forall u\in {\cal U}(d)
\end{eqnarray}
with the Shanon entropy $H(\mathscr{T}_j)=-\sum_i{p^{(j)}_i}\log{p^{(j)}_i}$.\\
\newline
The entropic bound informs us of Ailin's uncertainty in guessing Boris's testers' outcome. `Incompatibility' of testers here mean that Ailin cannot predict both Boris's outcome with certainty.

The similarity with the inequality $(\ref{ents})$ is obvious. 
In the following, we shall consider very specific cases of testers; one where we consider \textit{sets of testers} with each one having members with orthonormal states as inputs.
Let us thus define, towards this end, a \textit{complete set of testers} as follows.
\begin{definition1}
 A set $\mathfrak{T}=\{\mathscr{T},\mathscr{T}_2,...,\mathscr{T}_n\}$ with the inputs $|\psi_i\rangle$ for each tester $\mathscr{T}_i$ and common measurements, is complete if $\langle\psi_i|\psi_j\rangle=\delta_{ij}$ and $\sum_i^n|\psi_i\rangle\langle\psi_i|=I_{d\times d}$. 
 \end{definition1}
 
The use of a set of orthonormal states as above in testers is actually quite the standard in quantum process tomography \cite{jones,nel,sev}, and therefore represents a very practical scenario. We shall see how pairs of testers, each taken from a complete set, which saturate the minimal as well as maximal entropic bounds reveal certain interesting features of the unitary operators tested. This would be a reminiscence of the matter of entropic bounds of observables in distinguishing between quantum states.

\subsection{Trivial bound for a pair of testers}
\noindent Let us consider the case for the minimal bound; i.e. being equal to 0, which we shall refer to as the `trivial bound'. As the trivial bound implies that one can precisely predict the outcome of the testers, we can say that such testers are compatible. Trivial entropic bounds together with distinguishability of transformations for a common set of unitaries tested provide an interesting picture of compatibility of testers which can be further expanded to the equivalence notion.

\begin{definition1}
 \noindent Two testers $\mathscr{T}_1=(|\psi\rangle,\{|\chi_i\rangle\langle\chi_i|\})$ and $\mathscr{T}_2=(|\phi\rangle,\{|\zeta_j\rangle\langle\zeta_j|\})$ are defined as equivalent for a unitary $u$ if their probabilities distribution for the outcomes are equal; i.e. $\forall i,|\langle\chi_i|U |\psi_1\rangle|^2=|\langle\zeta_j|U |\phi_1\rangle|^2$ for some $j$ with $U=u\otimes I_d$.
 \end{definition1}
It is easy to show that the relation defined above is indeed an equivalence relation.
In what follows, we demonstrate how, given complete sets of testers like the above, testers derived from differing sets saturating the minimal entropic bounds can be equivalent. Let us first determine the condition of distinguishability of transformation by testers.

\begin{prop}\label{eigen}
Consider a tester $\mathscr{T}=(|\psi\rangle,\{|\chi_i\rangle\langle\chi_i|\})$ and $H(\mathscr{T},u_1)=H(\mathscr{T},u_2)=0$ for some $u_1$ and $u_2$.  Let $U_2^\dagger=(u_2\otimes I_d)^\dagger$ and $U_1=u_1\otimes I_d.$ $U_2^\dagger U_1$  is an eigenoperator for $|\psi\rangle$, if and only if $\mathscr{T}$ cannot distinguish between $U_1$ and $U_2$.

\end{prop}
\begin{proof}
Let $|\pi_{1}|=|\pi_{21}|=1$ be some global phase factors. $H(\mathscr{T},u_1)=0$ implies that $U_1|\psi\rangle=\pi_{1}|\chi_m\rangle$ for some $m$. If  $U_2^\dagger U_1$ is an eigenoperator for $|\psi\rangle$, then 
\[U_2^\dagger U_1|\psi\rangle=\pi_{21}|\psi\rangle\]
\[U_1|\psi\rangle=U_2\pi_{21} |\psi\rangle=\pi_1|\chi_m\rangle\]
Or
\[U_2|\psi\rangle=\pi_{21}^{-1}\pi_1|\chi_m\rangle\]
i.e. both $U_1$ and $U_2$ maps $|\psi\rangle$   to $|\chi_m\rangle$ (except for some global phase factor) and $\mathscr{T}$ cannot distinguish between them. 

On the other hand, if $\mathscr{T}$ cannot distinguish between $U_1$ and $U_2$, then both $U_1$ and $U_2$ maps $|\psi\rangle$   to $|\chi_m\rangle$ (except for some global phase factor, $|\pi_{1}|=|\pi_{2}|=1$). Hence
\begin{eqnarray}
 U_1|\psi\rangle=\pi_1|\chi_m\rangle ~;~
 U_2|\psi\rangle=\pi_2|\chi_m\rangle \nonumber\\ 
\end{eqnarray} 
and 
\[\pi_2^{-1}|\psi\rangle=U_2^\dagger|\chi_m\rangle\]
thus
\begin{eqnarray}
U_2^\dagger U_1|\psi\rangle=\pi_1\pi_2^{-1}|\psi\rangle ~;~
\end{eqnarray} 
i.e., $U_2^\dagger U_1$ is an eigenoperator for $|\psi\rangle$.
\end{proof}
The contrapositive of this proposition simply tells us that for $H(\mathscr{T},u_1)=H(\mathscr{T},u_2)=0$ as above, $\mathscr{T}$ \textit{can} distinguish between $U_1$ and $U_2$ if and only if $U_2^\dagger U_1$  is an \textit{not} eigenoperator for $|\psi\rangle$.

The distinguishability by $\mathscr{T}$ above also implies $\langle\psi|U_2^\dagger U_1|\psi\rangle=0$. It is worth noting that despite the use of $U_1$ and $U_2$, the issue of distinguishability here is obviously really between $u_1$ and $u_2$.

\begin{prop}\label{trivial}
Consider the complete set of testers $\mathfrak{T}_1$ and $\mathfrak{T}_2$ and let $\mathfrak{U}=\{U_1,...,U_D\}$ where $\forall m,n, U_m^\dagger U_n$ is neither an eigenoperator for any inputs from $\mathfrak{T}_1$ nor $\mathfrak{T}_2$. If $\mathscr{T}_i^{(1)}=(|\psi_i\rangle,\{|\chi_a\rangle\langle\chi_a|\})\in\mathfrak{T}_1$ and $\mathscr{T}_j^{(2)}=(|\phi_j\rangle,\{|\zeta_b\rangle\langle\zeta_b|\})\in\mathfrak{T}_2$ saturate the trivial entropic bound for any unitary operators in $\mathfrak{U}$, then the following statements follow;
\begin{description}
\item[S1] $\mathfrak{U}$ is a basis for some $D$ dimensional subspace of the set of operators $M(d,\mathbb{C})$
\item[S2]  $\mathfrak{U}$, and $\mathscr{T}_i^{(1)}$ and $\mathscr{T}_j^{(2)}$ are equivalent for any unitary in the subspace defined by $span(\mathfrak{U})$
\end{description}\end{prop}
\begin{proof}
Consider two unitaries $U_m=u_m\otimes I_d,U_n=u_n\otimes I_d\in \mathfrak{U}$. If 
\begin{eqnarray}
\forall i,j, H(\mathscr{T}_i^{(1)},u_x)+H(\mathscr{T}_j^{(2)},u_x)=0,
\end{eqnarray}
with $x=m,n$, then we obviously have $H(\mathscr{T}_i^{(1)},u_x)=H(\mathscr{T}_i^{(1)},u_x)=0$. Proposition  \ref{eigen} tells us that as $U_m^\dagger U_n$ is neither an eigenoperator for all inputs in $\mathfrak{T}_1$ nor $\mathfrak{T}_2$, thus $U_m$ and $U_n$ can be distinguished by all testers in both $\mathfrak{T}_1$ and $\mathfrak{T}_2$. 
Writing $\{|\psi_i\rangle\}$ as inputs  for testers from $\mathfrak{T}_1$, $\langle\psi_i|U_m^\dagger U_n|\psi_i\rangle=0$ and $\sum_i\langle\psi_i|U_m^\dagger U_n|\psi_i\rangle=Tr[U_m^\dagger U_n]=0$.\footnote{we could have chosen the case for $\mathfrak{T}_2$ with no loss of generality} Thus all elements in $\mathfrak{U}$ are orthogonal to one another and \textbf{S1} follows.\footnote{It is worth noting that \textbf{S1} can be established strictly based on the distinguishability of transformation by a single complete set of testers.}\\
\newline
Any unitary operator defined in the subspace of the basis $\mathfrak{U}$ can be written as $\sum_ka_kU_k$. Such a unitary acting on the input $|\psi_i\rangle$ or $|\phi_j\rangle$ gives $\sum_ka_kU_k|\psi_i\rangle\rightarrow\sum_ka_k|\chi_{(k)}\rangle$ and $\sum_ka_kU_k|\phi_i\rangle\rightarrow\sum_ka_k|\zeta_{(k)}\rangle$ respectively. We have used $|\chi_{(k)}\rangle$ and $|\zeta_{(k)}\rangle$ to denote the action of $U_k$ on $|\psi_i\rangle$ and $|\phi_i\rangle$ respectively. Thus the probability distributions, given by $\{|a_i|^2,i=1,...,d\}$, are the same for both the testers. This gives \textbf{S2}.

\end{proof}

\noindent The uncertainties in both cases when measurements are made are also identical, given by $\sum_k|a_k|^2\log_2({1/|a_k|^2})$. Thus, such sets of testers which saturate the trivial bound can be used to test any unitary in the subspace spanned by the basis $\mathfrak{U}$.

It is worth noting that as the unitaries tested only act on half the bipartite input state, proposition \ref{eigen} is really only relevant to ancilla-free type testers. This is easily seen as the eigenstates for any operator of the form $u\otimes I_d$ with $u\in\mathcal{U}(d)$ is a separable state in $\mathcal{H}_d\otimes\mathcal{H}_d$. Hence, for entangled inputs, proposition \ref{trivial} can be stated simply in terms of the saturation of the trivial bounds alone.

\subsection{Maximal Bounds}
%
\noindent We have seen how trivial bounds provide a picture of compatibility of testers. A maximal bound for a pair of testers on the other hand, gives us the notion of testers which are maximally incompatible. This is essentially the case where predictability of the outcome for one tester when testing a transformation implies complete uncertainty of the other. Let us assume the maximal bound to be $M$. Let the set $\mathfrak{T}_1$ and $\mathfrak{T}_2$ each be a complete sets of testers. Consider a pair of testers, $\mathscr{T}^{(1)}_i\in\mathfrak{T}_1$ and $\mathscr{T}^{(2)}_j\in\mathfrak{T}_2$ for any $i,j$, saturating the maximal entropic bound, $M$, for a common set of unitary transformations. As we can always find a set of unitary transformation, $\{u_1,...,u_D\}$, which minimises one complete set of tester say, $\mathfrak{T}^{(1)}_i$, such that $H(\mathscr{T}^{(1)}_i,u_n)=0$, then $H(\mathscr{T}^{(2)}_j,u_n)=M$. It is obvious now to note that given the number of outcomes for a tester is $D$, the value $M$ is achieved when all outcomes are equally likely, $H(\mathscr{T}_2,u_n)=\log_2{D}$.

Thus, a set $\mathfrak{U}$ which gives zero uncertainty for $\mathfrak{T}_1$ would result in maximal uncertainty for $\mathfrak{T}_2$. On the other hand, another set, say, $\mathfrak{U}'$ resulting in zero uncertainty for $\mathfrak{T}_2$ would give maximal uncertainty for $\mathfrak{T}_1$.
 \begin{prop}\label{MUUB}
Consider a pair of testers as above saturating the maximal entropic bound for the sets $\mathfrak{U}$ and $\mathfrak{U}'$. If both $\mathfrak{U}$ and $\mathfrak{U}'$ are bases for a common subspace for $M(d,\mathbb{C})$, these bases would be mutually unbiased to one another.
\end{prop}

\begin{proof}
Let us consider a pair of unitary transformations, $U_m=u_m\otimes I_d\in\mathfrak{U}$ and $U'_n=u'_n\otimes I_d\in\mathfrak{U}'$ both operating on $\mathcal{H}_d$. Let tester {\blu }$\mathscr{T}^{(1)}_{i}$ have its input state as $|\psi_i\rangle$ and the measurement operators $\{|\chi_k\rangle\langle\chi_k|\}$. 
We write thus,
\begin{eqnarray}
U_{m}|\psi_i\rangle=|\chi_{m_i}\rangle~,~U'_{n}|\psi_i\rangle=\dfrac{1}{\sqrt{D}}\sum\pi_{k_i}|\chi_{k_i}\rangle
\end{eqnarray}
with $|\pi_{k_i}|=1$ and 
\[
\langle\psi_i |U_{m}^\dagger U'_{n}|\psi_i\rangle=\dfrac{1}{\sqrt{D}}\pi_{m_i}
\]
If we consider the set of testers $\mathfrak{T}_1$, then the inputs are elements of the orthonormal basis $\{|\psi_l\rangle\}$, a similar treatment like the above gives
\[
\sum_l\langle\psi_l |U_{m}^\dagger U'_{n}|\psi_l\rangle=\dfrac{1}{\sqrt{D}}\sum_l\pi_{m_l}
\]
or 
\[
\left|\sum_l\langle\psi_l |U_{m}^\dagger U'_{n}|\psi_l\rangle\right |^2 =|\text{Tr}U_{T_2}^{\dag} U_{T_1}|^2=\dfrac{1}{D}\left |\sum_l\pi_{m_l}\right |^2
\]
Using the triangular inequality (and $\left|\pi_{m_l}\right |=1$ for any $l$),
\[\left |\sum_l\pi_{m_l}\right |^2\leq\left (\sum_l\left |\pi_{m_l}\right |\right)^2=D^2 \] and 
we have
 \begin{eqnarray}\label{U1U2}
0\leq |\text{Tr}(U_{m}^\dagger U'_{n})|^2\leq D.
\end{eqnarray}
We now consider, separately the case for $D=d^2$ and $D=d$.

\subsubsection{The case for $D=d^2$}
\noindent We set $D=d^2$ and we have
\begin{eqnarray}\label{u1u2}
|\text{Tr}(U_{m}^\dagger U'_{n})|^2=|\text{Tr}(u_m^{\dag} u_n')\text{Tr}(I_d)|^2\\\nonumber
=d^2|\text{Tr}(u_m^{\dag} u_n')|^2
\end{eqnarray}
Equations (\ref{U1U2}) and (\ref{u1u2}) give
{ \begin{eqnarray}
0\leq |\text{Tr}(u_m^{\dag} u_n')|^2\leq 1
\end{eqnarray}}
We now make use of the well known isomorphism between unitary operators, $u$ on $\mathcal{H}_d$ and (unnormalised) maximally entangled states, $|u\rangle\rangle$ in $\mathcal{H}_d\otimes \mathcal{H}_d$ along with the notation of  refs. \cite{d0,d1,d2},
\begin{eqnarray}
u\equiv \sum_{i}\sum_j \langle j|u|i\rangle |i\rangle|j\rangle=|u\rangle\rangle\in\mathcal{H}_d\otimes \mathcal{H}_d
\end{eqnarray} 
for some basis vectors $|i\rangle$,$|j\rangle$ of $\mathcal{H}_d$. With $|\langle\langle u_a|u_b\rangle\rangle|^2=\left|\text{Tr}(u_a^{\dagger}u_b)\right |^2$ and let $|\tilde u_a\rangle\rangle=|u_a\rangle\rangle/\sqrt{d}$ be a normalised maximally entangled state;
we now show that $|\text{Tr}(u_m^{\dag} u_n')|^2$ cannot be lesser than $1$. \\
\newline
If $|\text{Tr}(u_m^{\dag} u_n')|^2=|\langle\langle u_m|u_n'\rangle\rangle|^2\le 1$, then
\begin{eqnarray}
|\langle\langle \tilde u_m|\tilde u_n'\rangle\rangle|^2\le1/d^2
\Rightarrow\sum_i|\langle\langle \tilde u_m|{\tilde u_i'}\rangle\rangle|^2\le 1
\end{eqnarray}
However, $\sum_i|\langle\langle  \tilde u_m|{\tilde u_i'}\rangle\rangle|^2=1$ (as $\sum |\tilde u_i'\rangle\rangle\langle\langle \tilde u_i'|=I_{d^2}$), we can conclude that
the two unitary bases of $M(d,\mathbb{C})$, $\{u_1,...,u_{d^2}\}$ and $\{u_1',...,u_{d^2}'\}$ are such that, 
\[
|\text{Tr}(u_m^{\dag} u_n')|^2=1
\]
fulfilling the definition of MUUBs of eq.(\ref{MUUBdef}).
\subsubsection{The case for $D=d$}
The essential features for this part is the same as the previous; though there are few matters worth mentioning. The first is that, given the input state is $d$ dimensional and the projective measurements project onto $\mathcal{H}_d$, the testers can only discern between $d$ unitary operations. The other is that, we shall consider only the case where the sets of unitary operators (bases) tested form a common subspace. The requirement of bases of a common subspace was implicit in the previous as \textit{any} unitary operator belongs to some orthonormal basis of the same vector space $M(d,\mathbb{C})$, a scenario where the bases should span a common subspace.

Starting from equation (\ref{U1U2}), we consider the case for the unitary operators $v_m\in\mathcal{V}$ and $v_n'\in\mathcal{V}'$ 
\[
0\leq |\text{Tr}(v_m^{\dag} v_n')|^2\leq d
\]
it is not difficult to show that $|\text{Tr}(v_m^{\dag} v_n')|^2$ cannot be less that $d$. Using the same isomorphism as in the previous section, if $|\text{Tr}(v_m^{\dag} v_n')|^2\le d$ then 
\begin{eqnarray}
|\langle\langle \tilde v_m|\tilde {v_n'}\rangle\rangle|^2\le 1/d
\Rightarrow\sum_i|\langle\langle \tilde v_m|\tilde {v_i'}\rangle\rangle|^2\le 1~.
\end{eqnarray}
However,  $\sum_i|\langle\langle \tilde v_m|{\tilde v_i}'\rangle\rangle|^2=1$ thus $|\text{Tr}(v_m^{\dag} v_n')|^2=d$. It is instructive to note that $\{|\tilde v_1\rangle,...,|\tilde v_d\rangle\}$ and $\{|\tilde v_1'\rangle,...,|\tilde v_d'\rangle\}$ are orthonormal bases (of MES) which spans a common $d$ dimensional subspace of $\mathcal{H}_d\otimes\mathcal{H}_d$. Referring to \cite{arxiv}, these two bases are therefore mutually unbiased to one another. 
\end{proof}

\section{Simple Examples}

\noindent Ideally expressing the entropic bounds in terms of either the observables or input states (or both) independently of the unitary tested is ultimately a challenge. 
In this section we consider some simple cases.

Consider the case where the inputs for $\mathscr{T}_1$ and $\mathscr{T}_2$ are identical; $\mathscr{T}_1=(|\psi\rangle,\{|\chi_i\rangle\langle\chi_i|\})$ and $\mathscr{T}_2=(|\psi\rangle,\{|\zeta_i\rangle\langle\zeta_i|\})$. Then, 
\begin{eqnarray}
H(\mathscr{T}_1,u)+H(\mathscr{T}_2,u)=\sum_m |\langle\chi_m|\psi_u\rangle|^2\log{|\langle\chi_m |\psi_u\rangle|^2}\nonumber\\
+\sum_m |\langle\zeta_m |\psi_u\rangle|^2\log{|\langle\zeta_m |\psi_u\rangle|^2}\nonumber\\
=H(\{|\chi_i\rangle\langle\chi_i|\},|\psi_u\rangle )+H(\{|\zeta_i\rangle\langle\zeta_i|\},|\psi_u\rangle)
\end{eqnarray}
where $|\psi_u\rangle=u\otimes I_d|\psi\rangle$. Thus we see that for a given $u$, the entropic bound reduces to the entropic bound for observables in estimating a state $|\psi_u\rangle$. Hence for any $u$, we have
\begin{eqnarray}\label{maxeb}
H(\mathscr{T}_1,u)+H(\mathscr{T}_2,u)\ge-\log_2{\max\limits_{ij}|\langle\chi_i|\zeta_j\rangle|^2}
\end{eqnarray}
which reduces to the entropic bounds for the observables. In such a scenario, the maximal value for the entropic bound would be the case where eigenstates of the observables are from two MUBs.

An immediate example would be the testers with qubit inputs $\mathscr{T}_{0Z}=(|0\rangle,\mathbb{Z})$ and $\mathscr{T}_{0X}=(|0\rangle,\mathbb{X})$, where $\mathbb{X}$ and $\mathbb{Z}$ are the measurement bases corresponding to the Pauli observables. This pair actually saturate the entropic bound when testing the unitaries $I_2$ and $H=(I_2-i\sigma_Y)/\sqrt{2}$ (essentially the Hadamard operator where $\sigma_Y$ is the Pauli operator), which are from differing MUUBs. We shall see in the next section, how this can be used (see also \cite{js})  in constructing a bidirectional QKD protocol.

While the reduction to entropic bound for observables is true for identical inputs, cases for nonidentical inputs can be very different.
Consider the case $\mathscr{T}_{0Z}=(|0\rangle,\mathbb{Z})$ and $\mathscr{T}_{+X}=(|x_+\rangle,\mathbb{X})$. Despite the different observables, these testers saturate the trivial bound when distinguishing between the unitaries $I_2$ and $\sigma_Y$. Alternatively, consider the case $\mathscr{T}_{0Z}=(|0\rangle,\mathbb{Z})$ and $\mathscr{T}_{+Z}=(|x_+\rangle,\mathbb{Z})$. Despite identical observables, this pair never saturate the trivial bound.

\section{Application to QKD}

Prepare and measure QKD schemes like that of BB84 make use of measurement of observables in decoding. Thus, in a nutshell, the uncertainty principle guarantees the security of the shared secret between the legitimate parties. However, bidirectional QKD schemes (also referred to as two-way QKD) sees encoding as a unitary operation and thus the use of testers become of immediate interest. While earlier studies suggests that the security of such protocols to be based on the use of nonorthogonal states, later versions suggests it should be based on the inability to distinguish between the unitaries used \cite{chiri2, jspla,js}.

Let us review very quickly the bidirectional protocols and understand it in the context of testers. Referring to standard cryptographic communicating parties, Alice and Bob, the protocol begins with Bob sending to Alice a qubit selected from $2$ MUBs. Alice would then select a unitary transformation, either the identity operator or one that would flip the qubit to an orthogonal state before returning to Bob. Bob who measures the returned qubit in the same basis he prepared in would be able to infer Alice's operation by observing the evolution of his qubit. We see here the obvious fact that Bob actually, in every run of the protocol, is using a tester randomly picked from either one of two complete tester sets; $\mathfrak{T}_z=\{\mathscr{T}_{0Z}=(|0\rangle,\mathbb{Z}),\mathscr{T}_{1Z}=(|1\rangle,\mathbb{Z})\}$ or $\mathfrak{T}_x=\{\mathscr{T}_{+X}=(|x_+\rangle,\mathbb{X}),\mathscr{T}_{-X}=(|x_-\rangle,\mathbb{X})\}$. These testers saturate the trivial bound for the unitary operators used in the protocol, i.e. 
\begin{eqnarray}
H(\mathscr{T}_z,u)=H(\mathscr{T}_x,u)=0~,~\mathscr{T}_z\in \mathfrak{T}_z, \mathscr{T}_x\in \mathfrak{T}_x
\end{eqnarray}
with $u=I_2,i\sigma_y$. They are all equivalent to one another for the subspace of unitary operators spanned by the basis $\{I_2,i\sigma_y\}$. While this ensures Bob's decoding is perfect; it also means that Eve could use the same tester as Bob (or one equivalent to it) in his place before she replicates the transformation for Bob's tester. In the literature, this would be known as the Quantum Man in the Middle (QMM) attack where Eve hijacks Bob's qubit, sends her own to Alice and determine her encoding perfectly before acting on Bob's qubit with her new gained knowledge. Thus, a control mode (CM) where Alice randomly chooses to measure the received qubit in a basis ($\mathbb{X}$ or $\mathbb{Z}$) instead and comparing the results with Bob to alert them of Eve's presence becomes necessary. Unfortunately, this has effectively led to an execution of a `prepare-and-measure' protocol on the side-lines.

In designing such a bidirectional protocol, it is instructive to consider the use of the entropic bound. Motivated by the role of the uncertainty principle (entropic bounds for observables) in prepare and measure schemes \cite{pp,koashi}, we propose that Bob should use testers which saturates the maximal entropic bound of pairs of unitary testers of equation (\ref{maxeb}). For the sake of clarity, we shall begin with qubit based protocols. Using only $\mathscr{T}_{0Z}=(|0\rangle,\mathbb{Z})$ and $\mathscr{T}_{0X}(|0\rangle,\mathbb{X})$ to distinguish between $I$ and $H$, we come to a scenario where the state in the forward path is known to all, and privacy lies only in the inability to distinguish between two possible states from two differing MUBs in the backward path. This is effectively a B92 like protocol. Even if Bob uses the complete set of testers, $\mathfrak{T}_z=\{\mathscr{T}_{0Z}=(|0\rangle,\mathbb{Z}),\mathscr{T}_{1Z}=(|1\rangle,\mathbb{Z})\}$ and another, say, $\{\mathscr{T}_{0X}=(|0\rangle,\mathbb{X}),\mathscr{T}_{1X}=(|1\rangle,\mathbb{X})\}$, there is no real difference to the protocol's security as the states in the forward path is completely distinguishable. It is worth noting that as the testers of $\mathfrak{T}_x$ and $\mathfrak{T}_z$ are equivalent for the unitary operators in $span(\{I_2,i\sigma_y\})$, the bound 
 \begin{eqnarray}
H[\mathscr{T}_{nZ}=(|n\rangle,\mathbb{Z})]+H[\mathscr{T}_{mX}=(|m\rangle,\mathbb{X})]
\ge 1
\end{eqnarray}
where $m,n$ are either both from $\{|0\rangle,|1\rangle\}$ or both from $\{|x_+\rangle,|x_-\rangle\}$ holds when testing the unitaries $I_2$ and $H$. A proper modification can be to add more testers. Namely the sets  $\mathfrak{T}_x$ and $\{\mathscr{T}_{+Z}=(|x_+\rangle,\mathbb{Z}),\mathscr{T}_{-Z}=(|x_-\rangle,\mathbb{Z})\}$. The additional testers with inputs coming from differing MUBs would force some uncertainty upon Eve when attacking in the forward path. This is essentially the protocol discussed in ref.\cite{js}\footnote{the details of the protocol and Bob's decoding procedure is described in the ref.\cite{js}}.
An alternative would be to also increase the number of transformations that Alice can use. We shall describe this without limiting ourselves to the qubit scenario.

Let us consider Bob using the set of testers $\mathfrak{T}_B$ and $\mathfrak{T}_b$ which saturates the maximal entropic bound.  Thus, Alice may use a set of unitary encoding, $\mathfrak{U}_B=\{u_0^{(B)},...,u_{\mathfrak{D-1}}^{(B)}\}$ for which its elements can be distinguished perfectly by $\mathfrak{T}_B$ but maximises the uncertainty of $\mathfrak{T}_b$. Another set, $\mathfrak{U}_b=\{u_0^{(b)},...,u_{\mathfrak{D-1}}^{(b)}\}$ would on the other hand maximise $\mathfrak{T}_B$ but can be distinguished by $\mathfrak{T}_b$. The unitary operators $u_i^{(b)},u_i^{(b)}$ encodes the value $i$ for a $D$-ary key. To ensure that Eve would not be able to distinguish between the states used as the input of the tester (forward path), the tester sets should include equivalent testers with nonorthogonal states for input in each respective sets.

The protocol thus goes as follows: Bob selects a tester at random from either $\mathfrak{T}_B$ or $\mathfrak{T}_b$ and sends the tester's input to Alice. Alice would select an element from either set $\mathfrak{U}_B$ or $\mathfrak{U}_b$ to encode a $D$-ary digit. This is repeated for a large number of times and at the end of the protocol, Alice would declare publicly which sets were used (but not the specific unitary used). Bob would then discard the cases where his tester's uncertainty would have been maximised. Note that the protocol of \cite{jspla} is the case for $D=2$. 
Proposition \ref{MUUB} ensures the protocol makes use of MUUBs, $\mathfrak{U}_B$ and $\mathfrak{U}_b$ and $|\text{Tr}(u_i^{(b)\dag} u_j^{(B)})|^2$ takes on the value $1$ and $d$ for $D$ being $d^2$ and $d$ respectively. Given the fidelity $|\text{Tr}(u_i^{(b)\dag} u_j^{(B)})|^2/d^2$ of ref.\cite{acin}, this implies that when a unitary operator is chosen from one set, a guess of it being any operator coming from the other is equiprobable. 
Hence an adversary who would want to use an equivalent tester to Bob's to determine the encoding would only be able to maximise her information gain in half the time; with the other half experiencing maximal uncertainty. Obviously, a more involved analysis is required to properly address the most generic eavesdropping strategy in the framework of testers, thus a proper estimate of the security of such a protocol. While this is beyond the scope of this work, we conjecture it to be promising based on earlier works on bidirectional QKD using qubits making use of MUUBs \cite{jspla,js}.


\section{Conclusion}

In enlightening our understanding of nature, quantum mechanics has also  prescribed limitations on our ability to make precise measurements in distinguishing between quantum states. When it comes to distinguishing between unitary operators using testers, understandably, given that quantum states themselves are used as `test states' or probes, such limitations are carried over for pairs of testers.

In this work, we propose a quantitative formulation for the limits of knowledge one can have when testing unitary operators in terms of entropic bounds for a pair of testers used. We see how, when using a specific set of testers, namely complete set of testers, trivial and maximal entropic bounds reflect certain special properties of the unitary tested. Coupled with the issue of distinguishability of proposition \ref{eigen}, the trivial bound implies the operators tested form an orthogonal unitary operator basis. Maximal bounds on the other hand imply that the pair of unitary operators, for which one maximises one tester's uncertainty while the other minimises it (and vice versa), comes from two MUUBs. This is a reminiscence of the role of MUB in maximising the entropic bounds of observables. 
It is also interesting to note that the `similarities' extend even to the issue of application. The uncertainty principle has essentially led to the brith of quantum cryptography. Here we see how the uncertainty between testers play a similar role in the construction of quantum cryptographic schemes, specifically that of bidirectional QKD schemes. 

It is obviously interesting to have a more comprehensive understanding of the matter to include entropic bounds for more generalised testers with mixed states for input or POVMs for measurements. Or one may even imagine such bounds to exist between testers which test channels which are not necessarily unitary. We hope to address these in our future studies. 

\section{Acknowledgement}

J. S. S. would like to acknowledge financial support under the project FRGS19-141-0750 from the Ministry of Higher Education's Fundamental Research Grant Scheme and the University's Research Management Centre (RMC) for their support and facilities provided. J. S. S would also like to extend a special thanks to Aida (RMC) and her team for their kind assistance and encouragement. 

S. M acknowledges the funding from the European Union's Horizon 2020 research and innovation programme under grant agreement No 862644 (FET-Open project ``QUARTET").


\end{document}